\documentclass[a4paper,UKenglish, numberwithinsect]{lipics-v2016-mine}
 
\usepackage{microtype}

\title{Existence of k-ary Trees: Subtree Sizes, Heights and Depths}
\titlerunning{Existence of k-ary Trees} 

\author[1]{Akshar Varma}
\affil[1]{Dhirubhai Ambani Institute of Information and Communication
  Technology, Gandhinagar, India\\
  \texttt{akshar\_varma@daiict.ac.in}}
\authorrunning{Akshar Varma} 

\Copyright{Akshar Varma}


\newcommand{\nmts}{\textsc{NMTS}}
\newcommand{\nmtsk}{\textsc{NMTS-K}}
\newcommand{\kpwt}{\textsc{K-PwT}}

\newcommand{\sssk}{$\R(\I_{\S})$}

\newcommand{\npc}{NP-Complete}
\newcommand{\snpc}{Strongly \npc{}}
\newcommand{\pptrans}{pseudo-polynomial transformation}

\newcommand{\I}{\text{I}}
\newcommand{\R}{\text{E}}
\newcommand{\T}{\text{T}}
\renewcommand{\S}{\text{S}}
\newcommand{\D}{\text{D}}
\renewcommand{\H}{\text{H}}
\newcommand{\ITR}{\text{ITR}}

\newcommand{\Max}[1]{\textsc{Max[\textnormal{#1}]}}
\newcommand{\Length}[1]{\textsc{Length[\textnormal{#1}]}}
\newcommand{\f}{\textit{f}}

\newtheorem{problem}[theorem]{\textbf{PROBLEM}}
\newtheorem{myremark}[theorem]{Remark}

\newcommand{\psc}{\,}

\begin{document}

\maketitle

\begin{abstract}
  The rooted tree is an important data structure, and the subtree
  size, height, and depth are naturally defined attributes of every
  node. We consider the problem of the existence of a k-ary tree given
  a list of attribute sequences. We give polynomial time
  ($O(n\log(n))$) algorithms for the existence of a k-ary tree given
  depth and/or height sequences. Our most significant results are the
  Strong NP-Completeness of the decision problems of existence of
  k-ary trees given subtree size sequences. We prove this by
  multi-stage reductions from \textsc{Numerical Matching with Target
    Sums}. In the process, we also prove a generalized version of the
  \textsc{3-Partition} problem to be Strongly NP-Complete. By looking
  at problems where a combination of attribute sequences are given, we
  are able to draw the boundary between easy and hard problems related
  to existence of trees given attribute sequences and enhance our
  understanding of where the difficulty lies in such problems.
 \end{abstract}

\section{Introduction}
Rooted trees are important data structures that are encountered
extensively in Computer Science, especially in the form of
self-balancing binary trees. Attributes of nodes, like its subtree
size, height,
and depth are invariant under isomorphism of rooted trees
and are ubiquitous in the study of data structures and
algorithms. Heights are used in self-balancing trees (AVL~\cite{AVL}
or Red-Black~\cite{Red-Black}), depths in analyzing complexity of
computation trees, recursion trees and decision trees and subtree
sizes for finding order statistics in dynamic data
sets~\cite{CLRS}. The subtree size, height and depth of every node
in a rooted tree can be computed in time linear in the number of
nodes.  In this paper, we discuss the computational complexity of the
converse problem -- the existence
(realization) of a rooted k-ary tree given some of these attribute sequences.

The problems that we address are similar in flavor to those studied in
Aigner and Triesch~\cite{graph-invariants} who discuss realizability
and uniqueness of graphs given invariants. The most famous of such
problems is the Erdos-Gallai graph realization problem~\cite{Erdos-Gallai} (a
variant also addresses the realization problem for trees) which asks
whether a given set of natural numbers occur as the degree sequence of
some graph; polynomial time algorithms are known for this
problem~\cite{Havel,Hakimi}. Our problem can be considered a part of
the category of well researched problems of reconstruction of a
combinatorial structure from some form of partial information. Apart
from the Erdos-Gallai theorem we already mentioned, a lot of work has
been done on reconstruction of graphs from subgraphs~\cite{Ulam,
  Kelly, Harary1, Harary2, Harary3, Nash-Williams, Lovasz,
  ManvelGraphs}. Beyond graphs, the problem of
reconstruction of combinatorial structures like
matrices~\cite{ManvelMatrices} and trees~\cite{ManvelTrees} have also
been studied. More recent work has been done on reconstruction of
sequences~\cite{Dudik-Schulman} and on reconstruction of strings from
substrings~\cite{Acharya}. Bartha and Bursci~\cite{bartha-bursci} have
addressed the problem of reconstruction of trees using frequencies of
subtree sizes. While their paper focuses on the reconstruction of
unrooted trees given subtree sizes, we look at the existence of rooted
trees with given attribute sequences.

\sloppy Given a rooted tree \T{}, information \I(\T) about the
attributes of the tree \T{} can be constructed using various
combination of attribute sequences. Given some such information \I, we look at
the existence problem \R(\I), which asks whether there is a k-ary
tree \T{} such that $\I=\I(\T)$. We use the letters \S{}, \H{}, and
\D{} to refer to subtree size, height and depth attributes
respectively. These attributes can be used individually or in
combination as in these examples:

\begin{example}\label{subtree} We
use the notation like $\R(\I_\S)$ for the existence problem given
only, say,
the subtree sizes sequence\footnote{We use the term sequence (borrowed from the
  terminology in the Erdos-Gallai theorem) throughout the paper since
  these attributes are generally computed and used in either
  non-decreasing or non-increasing order. It should be noted that when
  such an order is not imposed, these are multisets. Nonetheless, we
  maintain the use of the term sequence for consistency.}.
For example: Given a sequence of subtree sizes,
  $\I_\S=\{1,2,3,1,1,3,7,1,9\}$,
  does there exist a tree T such that $\I_\S=\I_\S(T)$?\end{example}

\begin{example}\label{setOfTuples}
We use the notation like $\R(\I_{\S,\D})$ to refer
to the existence problem given, say, synchronized subtree sizes and depths (all attributes of a node are
associated with each other as tuples).
For example: Given synchronized information of (subtree size,
  depth), $\I_{\S,\D}=\{(1,5),(2,4),(3,2),(1,3),(1,3),(3,2),(7,1),(1,1),(9,0)\}$, does
  there exist a tree T such that $\I_{\S,\D}=\I_{\S,\D}(T)$?\end{example}%

\begin{example}\label{tupleOfSets} We use the notation like $\R(\I_\S,\I_\D)$ when there is no synchronization and
  just two (or more) sequences are given. For example: Given
  asynchronized information list of subtree sizes and depths,
  $\I_\S=\{1,2,3,1,1,3,7,1,9\}, \I_\D=\{5,4,2,3,3,2,1,1,0\}$,
  does there exist a tree T such that $\I_\S=\I_\S(T)$ and $\I_\D=\I_\D(T)$?\end{example}

Problems containing only height and/or depth sequences are shown to have
$O(n\log(n))$ algorithms for deciding the existence of k-ary
trees. Our most significant results are the proof that all existence
problems containing subtree size sequences are Strongly
NP-Complete for k-ary trees. We prove the Strong NP-Completeness using reductions from
the \textsc{Numerical Matching with Target Sums} problem. The
reduction is performed in multiple stages during which we also prove a
generalized version of the \textsc{3-partition} problem to be Strongly
NP-Complete. We then proceed to provide slightly modified yet similar existence
problems (for example, existence of certain sub-classes of trees)
which are polynomially solvable. We attempt to draw the boundaries
separating the NP-Complete problems from the easy problems, focusing
on how changing the attribute, adding restrictions or providing more
information change the computational complexity.

Section \ref{sect:def} contains the basic definitions, notation and
conventions. The most important results are presented in Section
\ref{sect:proofs}, which contains the proofs for the Strong
NP-Completeness of problems related to subtree sizes. Section
\ref{sect:sssk-positive} continues further discussion on subtree sizes
and contains algorithms for some sub-classes of trees for which the
problem can be solved in polynomial time. This is followed by the
Section \ref{sect:HDstuff}, detailing the analyses of the height and
depth sequences. Section \ref{sect:combined} contains details of
sequences given in combination and some discussion about the
difficulty of these problems. Section \ref{sect:conclusion} has
concluding remarks and possible directions for future work.

\section{Preliminaries}\label{sect:def}
   \begin{definition}[k-ary Tree]\label{def:kary} A rooted tree in which every node can have at
     most K children.\end{definition}

   \begin{definition}[Subtree size]\label{def:subtreeSize} The number of nodes
     in the subtree rooted at a node (including itself) is known as
     the subtree size of that node. The subtree size of a node can
     also be defined recursively as being one greater than the sum of the subtree sizes
     of its children.\end{definition}

   \begin{definition}[Height of a node]\label{def:height} The height
     of a leaf node is zero. The height of every other node is one
     more than the maximum of the heights of its children.\end{definition}

  \begin{definition}[Depth of a node]\label{def:depth} The depth of a
    node is the number of edges in the path from that node to the
    root.\end{definition}

  \begin{definition}[Levels in a tree]\label{def:level} The set
    of nodes at a particular depth forms the level at that
    depth.\end{definition}

 \begin{definition}[Complete Trees]\label{def:complete} $K$-ary trees in which every level except
  possibly the last are filled and all nodes in the last level are
  filled from the left. Given the number of nodes $n$, this is a unique
  tree and is represented as $T^c(n)$.\end{definition}

 \begin{definition}[Full Trees]\label{def:full} $K$-ary trees in which
   every node has 
   exactly $K$ or $0$ children.\end{definition}

 The remaining definition are due to Gary and Johnson~\cite{garyjohn, strongNP}.
 \begin{definition}[\Max{I}] \Max{I} is the magnitude of the
   maximum number present in an instance $I$ of any decision problem. \end{definition}

  \begin{definition}[\Length{I}] \Length{I} is the number of
    symbols required to represent an instance $I$ of any decision problem. \end{definition}

  \begin{definition}[Strongly NP-Completeness]\label{def:snpc}
    For a decision problem $\Pi$, we define $\Pi_p$ to denote the
    subproblem of $\Pi$ obtained by restricting $\Pi$ to only those
    instances that satisfy $\Max{I} \leq p(\Length{I})$, where $p$ is a
    polynomial function. The problem $\Pi$ is said to be NP-Complete
    in the strong sense or Strongly NP-Complete (SNPC) if $\Pi$ belongs to
    NP and there exists $p$ for which
    $\Pi_p$ is NP-complete.
  \end{definition}

  \begin{definition}[Pseudo-polynomial transformation]\label{def:pseudo}
    A pseudo-polynomial transformation from a source problem to a
    target problem is a transformation \f{} from any instance of the source
    problem to an arbitrary instance of the target problem satisfying the following
    conditions:
    \begin{itemize}
    \item $f$ should be computable in a time
      polynomial in \Max{I} and \Length{I}.
    \item $\Length{I} \leq p(\Length{\f(I)})$ for some polynomial $p$.
    \item $\Max{\f(I)} \leq p'(\Max{I}, \Length{I})$ for some polynomial $p'$.
    \end{itemize}
  \end{definition}

  \begin{myremark}\label{remark-for-proof}
    To prove problems to be \snpc{}, one needs to have a \pptrans{} from
  a problem already known to be \snpc{}. In practice this is the same as a
  polynomial transformation used to prove problems NP-Complete,
  with the added condition that the maximum integer in the constructed
  instance needs to be polynomially bounded in the maximum integer in and
  the length of, the instance from which we are making the transformation.
\end{myremark}

\section{Existence of k-ary Trees given Subtree Sizes Sequence: \sssk{}}

Given a tree, finding the subtree sizes of all its nodes is a linear
time problem. One might intuitively expect that the converse problem is
also easy. After all, by definition, all that needs to be assured is
that the sum of the children's subtree sizes is one less than the
parent's subtree size. While intuitively this may seem so, the \sssk{}
problem has been found to be difficult to be solve for binary
trees~\cite{manuByKS}. Top-down, bottom-up and dynamic programming
approaches were tried but all yielded exponential time
algorithms. This difficulty prompted a search for an NP-Completeness
reduction.

We prove the \sssk{} problem \snpc{} in a series of reductions
starting from the \textsc{Numerical Matching with Target Sums} (\nmts{})
problem (\snpc{} by Theorem~\ref{theorem:nmts}) to the \textsc{Numerical
Matching with Target Sums using K-sets} (\nmtsk{}) problem (in
Section~\ref{sect:prove-nmtsk}) to the \textsc{K-Partition with Targets}
(\kpwt{}) problem (in Section~\ref{sect:prove-kpwt}) to finally the
\sssk{} problem (in Section~\ref{sect:prove-sssk}).

\subsection{Proofs of Strong NP-Completeness}\label{sect:proofs}

In this section we prove the Strong NP-Completeness of \sssk{} via a series of reductions.

\begin{theorem}[Due to Garey and Johnson~\cite{garyjohn}]\label{theorem:nmts} The
  \nmts{} problem stated below is strongly NP complete:

  Given disjoint sets $X$ and $Y$ each containing $m$ elements, a size
  function $s:~X\cup Y \mapsto \mathbb{Z^{+}}$, and a target vector
  $B=(b_1,\dots,b_m) \;\in \mathbb{N}^m$ with positive integer
  entries, can $X \cup Y$ be partitioned into $m$ disjoint sets
  $A_1,A_2,\dots,A_m$, each containing exactly one element from each
  of $X$ and $Y$, such that, $\sum_{a \in A_i}s(a) = b_i$, for
  $1 \le i \le m$?
\end{theorem}
\noindent
\nmts{}~$(X, Y, s, B, m)$ refers to an instance of the \nmts{} problem characterized by the sets $X$ and $Y$, a size
function $s$, the target vector $B$ and the cardinality of the target
vector $m$.

\subsubsection{\nmtsk{} is \snpc}\label{sect:prove-nmtsk}
The \nmtsk{} problem is proved \snpc{} by reduction from \nmts{}.

\begin{problem}[\nmtsk{}]\label{prob:nmtsk}
  Given $K \geq 2$ disjoint sets $X_i$ each containing $m$ elements, a
  size function $s:~\bigcup X_i \mapsto \mathbb{Z^{+}}$, and a target
  vector $B=(b_1,\dots,b_m) \;\in \mathbb{N}^m$ with positive integer
  entries, can $\bigcup X_i$ be partitioned into $m$ disjoint sets
  $A_1,A_2,\dots,A_m$, each containing exactly one element from each
  of $X_i$, such that, $\sum_{a \in A_i}s(a) = b_i$, for
  $1 \le i \le m$?
\end{problem}
\noindent
\nmtsk{}~$(K, X_i, s, B, m)$ is an instance of the \nmtsk{} problem characterized by the integer $K$, the $K$ sets $X_i$, a
size function $s$, the target vector $B$ and the cardinality of the
target vector $m$.

\begin{proof}
  The \nmtsk{} problem is in NP since given a candidate partition
  $A_i$, we only need to verify that, $\sum_{a \in A_i}s(a) = b_i$,
  for $1 \le i \le m$. We now construct an instance
  \nmtsk{}$(K, X_i, s', B', m')$ of \nmtsk{} problem from an instance
  \nmts{}$(X, Y, s, B, m)$ of the \nmts{} problem using the following
  transformation for $K \geq 3$ since for $K=2$, the \nmtsk{} problem
  is the \nmts{} problem. Note that this is a polynomial
  transformation since computing Equations \ref{eq:nmtsk-xi},
  \ref{eq:nmtsk-size} and \ref{eq:nmtsk-bi} can be done in polynomial
  time.

  \begin{gather}
    m' = m,\; X_1 = X,\; X_2 = Y\\
    X_i \text{ are disjoint sets such that } |X_i| = m' \text{ for } 3
    \leq i \leq K\label{eq:nmtsk-xi}\\
    s'(x)=\begin{cases}s(x) & x \in X \cup Y\\
      1 & \text{otherwise}\\
      \end{cases}\label{eq:nmtsk-size}\\
    B'=(b'_1, b'_2, \dots, b'_m) \text{ where }  b'_i=b_i+ K-2, \text{
    } \forall b_i \in B\label{eq:nmtsk-bi}
  \end{gather}

  We now prove that a YES instance of the \nmtsk{} problem occurs iff
  a YES instance of \nmts{} occurs. Every partition for the \nmts{}
  problem is associated with a partition for the \nmtsk{} problem. We
  denote the elements of $X_i \text{ for } i \geq 3$ as
  $x_{ij}, 1 \leq j \leq m$ and let $A_i$ be the partition for the
  \nmts{} problem. The associated partition for the \nmtsk{} problem $A'_i$, is defined as follows:
  $A'_i=A_i \cup \{x_{ji}|3 \leq j \leq K\}$. This association
  immediately provides us with the equality: $\sum_{x \in A'_i} s'(x)
  = (\sum_{x \in A_i} s(x)) + (K-2)$ which we compare with the
  relation $b'_i=b_i+(K-2)$ from Eq.~\ref{eq:nmtsk-bi}. We get that
  $\sum_{x \in A'_i} s'(x)=b'_i$ and $\sum_{x \in A_i} s(x)=b_i$ either
  happen simultaneously or not at all. Thus, this association ensures
  that this is a valid transformation.

  The maximum number in the constructed instance is either the maximum
  size from $X$ and $Y$ or $K-2$ added to the maximum number from the
  target $B$; both of which are polynomially bounded in the maximum
  integer in, and the length of, the \nmts{} instance. This, along with  Remark~\ref{remark-for-proof} proves that \nmtsk{} is \snpc{}.
\end{proof}

\subsubsection{\kpwt{} is \snpc{}}\label{sect:prove-kpwt}

The \kpwt{} problem is \snpc{} by reduction from the
\nmtsk{} problem. This problem can be regarded as a generalization
of the \textsc{3-Partition} problem where we are looking for a
partition into K-sets and there are multiple targets to be reached
instead of a single target.

\begin{problem}[\kpwt{}]
  Given a set $X$ with $|X|=Km$, $K\geq 2$, a size function $s: X \mapsto
  \mathbb{Z^{+}}$ and a target vector $B=(b_1,\dots,b_m) \;\in
  \mathbb{N}^m$ with positive integer entries, can $X$ be partitioned into $m$ disjoint sets
  $A_1,A_2,\dots,A_m$, each containing exactly $K$ elements, such that, $\sum_{a \in A_i}s(a) = b_i$, for $1 \le i \le m$?
\end{problem}
\noindent
\kpwt{}~$(K, X, s, B, m)$ is an instance of the \kpwt{} problem characterized by the set $X$, an integer $K$, a
size function $s$, the target vector $B$ and the cardinality of the
target vector $m$.

\begin{proof}
  The \kpwt{} problem is in NP since given a particular candidate partition
  $A_i$, we only need to verify that, $\sum_{a \in A_i}s(a) = b_i$,
  for $1 \le i \le m$. We now construct an instance \kpwt{}~$(K', X, s', B', m')$ of \kpwt{}
  problem from an instance \nmtsk{}~$(K, X_i, s, B, m)$ of
  the \nmtsk{} problem in the following manner.

  $M$ is polynomially bounded by the maximum integer in the \nmtsk{}
  instance.

  \begin{gather}
    M =KmM' \;\text{ where, } M'=\max\big(\{s(x_i) | x_i \in X_i\} \cup \{b_i|b_i \in
    B\}\big)\label{eq:kpwt-max}\\
    K'=K, \text{ } X= \bigcup X_i\label{eq:kpwt-sets}\\
    \text{For}\;  1 \leq i \leq K \text{ and } \forall x_j \in X_i: \; s'(x_j)=s(x_j)+M^i\label{eq:kpwt-size}\\
    B'=(b'_1, b'_2, \dots, b'_m) \text{ where } b'_j=b_j+\sigma, \; \forall b_j \in B \;
    \text{ and } \sigma =\sum_{i=1}^K M^{i}\label{eq:kpwt-target}
  \end{gather}

  The transformation is polynomial since equations \ref{eq:kpwt-max}
  to \ref{eq:kpwt-target} are polynomial-time computable. Now we show
  that a YES instance of the \kpwt{} problem occurs iff a YES instance
  of \nmtsk{} occurs. For ease of exposition, for the rest of the
  proof, we write all the numbers in \kpwt{}~$(K', X, s', B', m')$ in
  base $M$. We make three remarks, the first: $\sigma$ is a $K+1$ digit
  number with a 1 in all its digits except the rightmost or $0^{th}$
  digit (Eq.~\ref{eq:kpwt-target}). The second, that every number
  $s'(x)$ has a 1 as its $i^{th}$ digit, $s(x)$ in its rightmost
  digit\footnote{Follows from Eq.~\ref{eq:kpwt-size} and $M$ being much greater in magnitude than any number in
  the \nmtsk{} instance.} and 0 elsewhere. The third, a partition
$A_j$ for the \nmtsk{} instance ($\bigcup X_i$) is also a partition for the \kpwt{}
instance ($X$), irrespective of whether either of them solve the respective
problems or not. We'll prove first that if $A_j$ is a partition that
solves the \nmtsk{} problem, then it also solves the \kpwt{} problem.

Let $A_j$ be a partition that solves the \nmtsk{} problem. Using the
same partition and Eq.~\ref{eq:kpwt-size} and \ref{eq:kpwt-target}, we get that
$\sum_{x \in A_j} s'(x)=\sum_{x \in A_j} s(x) +
\sum_{i=1}^{K}M^i=\sum_{x \in A_j} s(x) + \sigma=b_j +
\sigma=b'_j$. This proves that if $A_j$ solves the \nmtsk{} problem, then it
also solves the \kpwt{} problem.

To prove the converse, let $A_j$ be a partition that solves the \kpwt{}
problem. We know that $\sum_{x \in A_j} s'(x) = b_j + \sigma$, which
implies that $\sum_{x \in A_j} s(x) + \sum_i\sum_{x\in{}X_i\cap{}A_j}M^i = b_j +
\sigma$ (from Eq.~\ref{eq:kpwt-size}). This in turn implies that
$\sum_{x \in A_j} s(x) = b_j$ and $\sum_i\sum_{x\in{}X_i\cap{}A_j}M^i = \sigma = \sum_{i=1}^{K}M^i$ since $s(x)$
does not contribute to $\sigma$ (from the first two remarks and
Eq.~\ref{eq:kpwt-target}). Given $\sum_i\sum_{x\in{}X_i\cap{}A_j}M^i
= \sum_{i=1}^{K}M^i$, equating the coefficients of the powers of $M$,
we get that $|X_i\cap{}A_j|=1 \text{, } \forall i,j$ which says that
every set in the partition contains exactly one element from each of
the sets $X_i$. We already
know from the earlier equations that $\sum_{x \in A_j} s(x) = b_j$. Thus, the partition $A_j$ is a solution to the \nmtsk{} problem as well.

The maximum integer in the \kpwt{} instance created by the
transformation, $\sigma + \max(b_1,\dots,b_m)$, is bounded (from
Eq. \ref{eq:kpwt-max}) by a polynomial in the maximum integer in, and
the length of, the \nmts{} instance, which by Remark~\ref{remark-for-proof} makes this is
a \pptrans{}. Thus, \kpwt{} is \snpc{}.
\end{proof}

\subsubsection{\sssk{} is \snpc{}}\label{sect:prove-sssk}

We prove that the \sssk{} problem is \snpc{} by reduction from the
\kpwt{} problem. We use a subclass of the \kpwt{} problem, $\Pi_p$, such that
$\Max{I} \leq p(\Length{I}),\; \forall I \in \Pi_p$. By definition
\ref{def:snpc}, $\Pi_p$ is \npc{}. The \sssk{} problem:

\begin{problem}[\sssk{}]
  Given a sequence $\I_\S$ does there exist a $k$-ary tree $\T$, such
  that $\I_\S=\I_\S(\T)$?
\end{problem}

The \sssk{} problem is characterized by the set $S$ and the integer $k$. We refer to such an
instance as \sssk{}$(S, k)$.

\begin{proof}
  It is easy to see that the \sssk{} problem is in NP. Given $\T$, one
  only needs to compare $\I_\S$ with $\I_\S(T)$ to see if the given
  tree realizes that sequence. We now construct an instance
  \sssk{}$(S, k)$ from an instance of the \kpwt{} problem,
  \kpwt{}~$(K, X, s, B, m)$, with $k=K$.

  We define a number $M$ which is a power of K, is much greater in
  magnitude than any of the other numbers in the problem, and is
  polynomially bounded by the maximum integer in the \kpwt{}
  instance\footnote{$K^{\lceil\log_{K}\alpha}\rceil \leq
    K^{1+\log_{K}\alpha} = K\alpha. \therefore$ it is polynomially bounded.} (Eq.~\ref{eq:sssk-max} and \ref{eq:sssk-M}). We also define $m'$ and $m''$ such
  that $m+m'$ and $m+m''$ are powers of $K$.(Eq.~\ref{eq:sssk-m'm''}).
  \begin{gather}
    M_1=\max\big(\{s(x_i) | x_i \in X\} \cup \{b_i|b_i \in B\}\big)\label{eq:sssk-max}\\
    M=K^{\lceil\log_k M_2 \rceil}, \text{ where } M_2 =KmM_1\label{eq:sssk-M}\\
    m'=K^d-Km,\; m''=K^{d-1}-m=m'/K \text{, where }
    d=\lceil\log_K(Km)\rceil\label{eq:sssk-m'm''}
  \end{gather}
  
  We make the sequence $S$ for the \sssk{} instance using four
  ``component'' sequences, namely the ``child component'' C, the
  ``parent component'' P, the ``grandparent component'' G and the
  ``descendant component'' D:
    \begin{gather}
    S=C \cup P \cup G \cup D    \text{ where,}\label{eq:sssk-S}\\
    C=C' \cup C'',\; C'=\big\{s(x)+M \big|\psc x \in X\big\},\; C''=\big\{\overbrace{M,\psc
    \dots,\psc M}^{m' \mathrm{times}}\big\}\label{eq:sssk-C}\\
    P=P'\cup P'',\; P'=\big\{b_i+KM+1 \big|\psc b_i \in B \big\},\; P''=\big\{\overbrace{KM+1,\psc
    \dots,\psc KM+1}^{m'' \mathrm{times}}\big\}\label{eq:sssk-P}\\
    G=\bigcup_{i=0}^{d-2}l_i, \text{ where } l_i \text{ are ``levels''
      defined later in the text.} \label{eq:sssk-G}\\
    D=\bigcup_{i=1}^{K^d}D_i, \text{ where } D_i=\bigcup \I_\S(T^c(c_i)),\;
    \forall c_i\in C\label{eq:sssk-D}
  \end{gather}

  The ``child component'' $C$ is the union of the sets $C'$ and $C''$.
  $C'$ is in one-to-one correspondence with the set $X$, using the
  sizes of elements from $X$ with $M$ added to them. The set $C''$ is
  used to make the cardinality of set $C$ to be a power of $K$ using
  elements of value $M$. The ``parent component'' $P$ is the union of
  the sets $P'$ and $P''$. $P'$ is in one-to-one-correspondence with
  $B$ but has been modified to accommodate the changes made to sizes
  of elements of $X$ while making $C'$. $P''$ is used to make the
  cardinality of the set $P$ to be a power of $K$ using elements of
  value $KM+1$.
  
  We construct the ``grandparent component'' in ``levels''. The lowest
  level $l_{d-2}$ is constructed from $P$, by arbitrarily taking
  blocks of $K$ elements, adding them all up and incrementing the
  result by one. Formally, we order the elements in $P$ arbitrarily as
  $P_1, P_2, \dots, P_{K^{d-1}}$ and then let
  $l_{d-2}=\{l_{d-2,i} \;|\; l_{d-2,i}=1+\sum_{j=1}^K P_{(i-1)K+j},
  \; 1 \leq i \leq K^{d-2} \}$. Other levels $l_{d-i}$ are constructed
  in a similar manner from levels $l_{d-i+1}$. This is continued until
  $l_0$ which has only one element\footnote{The ``parent component''
    was padded with elements until the number of elements became a
    power of $K$. Since at each level the number of elements gets
    reduced by exactly a factor of $K$, eventually exactly one element
    will remain.}. The element in $l_0$ would be the largest number in
  the final instance.

The ``descendant component'' is constructed by using $T^c(c_i)$, the
subtree sizes sequences of complete trees on each of the elements
$c_i \in C$ (Refer to definition \ref{def:complete}). That is, for
each such $c_i$, we make a complete k-ary tree on $c_i$ nodes and find
its subtree size sequence $\I_\S(T^c(c_i))$. Let these sequences be
labeled $D_i$. The descendant component is $D=
\bigcup_{i}D_i=\bigcup_i \I_\S(T^c(c_i))$.

This is a polynomial transformation since each element from the
\nmtsk{} instance is being used only once and each time a simple
addition is done to get $P$ and $C$. There are a logarithmic number
($O(d)$) of $l \in G$ and each is computed in polynomial time. $D$ is
made up of $\I_\S(T^c(c_i))$ for each $c_i \in C$ which has a
polynomial number of elements and computing this for each element can be done in
polynomial time\footnote{We are reducing from $\Pi_p$, which ensures
  that the elements in the $C$ are polynomially bounded by
  $\Max{I}\; \forall I \in \Pi_p$ and thus, the subtree sizes
  sequences on complete trees on these number of nodes can be computed
  polynomially.}. Thus, every component can be computed in polynomial
time and so the whole \sssk{} instance can be constructed in
polynomial time.

Now we show that a YES instance of the \nmtsk{} problem occurs iff a
YES instance of \sssk{} occurs. By construction, all elements in $G$
and $P$, together, will form a k-ary tree with the elements in $P$ as
the leaves. Also $C$ and $D$ will make a forest of k-ary trees with
each element from $C$ as a root of one of the trees in the
forest. Since $C'$ is in one-to-one correspondence with $X$ and $P'$
is in one-to-one correspondence with $B$, if there is a partition,
$C'$ gets partitioned accordingly and these become the children of
elements of $P'$. $C''$ can be arbitrarily partitioned and made the
children of elements of $P''$. This will provide the remaining edges
to make a tree.

Now we need to prove that if there is a tree, then there is also a
partition. For this, we only need to prove that the set of children of
$P'$ is equal to $C'$. To prove this, it is sufficient to show that
the elements of $P'$ and the elements of $C'$ occur in consecutive
levels in any tree. We use equations \ref{eq:sssk-C} to
\ref{eq:sssk-D} to prove this. The element in $l_0$ is the largest
element and will necessarily have to be the root. This will be
followed by the elements from $l_1$ since no other elements are large
enough to reach the element in $l_0$. Continuing this argument it is
clear that the $l_i \in G$ will always appear in consecutive levels in
any tree and that $P$ will follow immediately below these levels. Now,
since no element from $p \in P$ will be a child of any $p' \in P$,
elements from either $C$ or $D$ will be needed to make child nodes of
elements in $P$. But we note that elements in $D$ will all be less
than $M/K$ in value which will not be enough to reach elements in $P$
thus necessitating that all children of elements of $P$ come from
$C$. We note that elements from $C'$ can not be children of elements
from $P''$ and so the set of children of $P''$ will be equal to
$C''$. Since a value of the order of $KM$ has to be reached for
elements in $P'$ and all elements in $C'$ are of the order of $M$, all
the elements from $C'$ will be used. Thus, if there is a tree, then it
will have the elements from $P'$ and $C'$ in consecutive levels and
therefore have a partition.

This transformation is sufficient to prove that \sssk{} is \npc{}. For
it to be \snpc{}, a subproblem of \sssk{} in which the maximum integer
is bounded by the length of the instance has to be proven \npc{}. We
note that the maximum integer in this reduction is the element in
$l_0$. We have already argued how this integer is polynomially bounded
by the length of the problem. This in turn proves that this reduction
is also sufficient to prove that \sssk{} is \snpc{}.
\end{proof}

\begin{corollary}[\sssk{} for full k-ary trees]\label{coro:full-sssk}
The \sssk{} problem for full k-ary trees (existence of full k-ary
trees given $\I_\S$) is also \snpc.
\end{corollary}
In the proof (Ref. Section \ref{sect:prove-sssk}) of the \sssk{}
problem for k-ary trees, we are reducing (when it is a YES instance)
the \kpwt{} instance to a full k-ary tree, except possibly in $D$. $D$ is being made of complete trees on elements of $C$. A simple
inductive proof is enough to show that changing every element
$c \in C$ to the form $Kc + 1$ would be enough to make all of the
complete trees in $D$ to also be full trees. This (along with similar
changes to $P$) would allow the same reduction to be used to prove
\sssk{} to be \snpc{} for full k-ary trees as well.


\subsection{Sub-classes realizable in polynomial time}\label{sect:sssk-positive}
While the existence problems for a k-ary tree and for the full k-ary
are \snpc{}, some sub-classes can be realized in polynomial time.

\begin{enumerate}
\item Complete Trees: The complete tree $\T^{c}(n)$ on a given number
  of nodes $n$ has a unique structure. Given a sequence $\I_\S$, we
  construct $\T^{c}(|\I_\S|)$ and simply check if
  $\I_\S = \I_\S(\T^{c})$.
\item Degenerate Trees (A tree which is just
    a path): $\I_\S$ must contain exactly one instance of every
  number from 1 to $|\I_\S|$.
\end{enumerate}

These along with corollary \ref{coro:full-sssk} show that the
structure of the tree we are trying to realize plays an important role
in deciding the complexity of the problem. The full k-ary tree which
allows flexibility in terms of the structure of the tree is \snpc,
while for the more rigidly structured complete and degenerate trees,
the problem is trivial.

\section{Height and Depth}\label{sect:HDstuff}

Depths, heights and subtree sizes can be recursively defined on the
basis of the attribute values of some neighbor nodes and given a tree,
each of the sequences $\I_\S$, $\I_\H$, $\I_\D$ can be computed in
linear time. We have seen that realizing trees given the subtree sizes
sequence is NP-Complete but when we consider sequences of other
attributes, we see that they are much easier to solve. We now provide
$O(n\log(n))$ algorithms for determining the existence of trees given
height or depth sequences.

\subsection{Depth}\label{sect:depth}
We know that in $K$-ary trees, every node can have at most $K$
children. And by definition, depth values give the level at which a
node is present. Let us define $C_d$ to represent the number of times
the value $d$ occurs in the sequence $\I_\D$. If we have built a k-ary
tree down to $d$ levels then the next level can accommodate at most
$KC_d$ nodes. Hence, if for every $d$ from $0$ to $d_{max}-1$, the
condition $C_{d+1} \le KC_d$ holds, then a tree can be constructed. If
the condition fails at any point then we get a proof that no tree
exists.

\begin{gather}\label{eq:depth-condition}
  C_{d+1} \leq KC_d, \text{ for } 0 \leq d \leq d_{max}-1
\end{gather}

\subparagraph{Sub-classes}
\begin{itemize}
\item Complete tree: Eq.~\ref{eq:depth-condition} is modified to
  be $C_{d+1}=KC_d$  for all $d \neq d_{max}-1$.
\item Full $K$-ary tree: Along with the condition in
  Eq.~\ref{eq:depth-condition} we add the extra condition that $K$
  divides $C_d$ for all $d$. That is, $K|C_d \quad \forall d$.
\item Degenerate Tree (A path): Exactly one count of each value must be
  present.
\end{itemize}

It is easy to see that these will all result in single pass algorithms
once we have sorted the sequence and thus their running time is of the
order of $O(n\log(n))$ where $n$ is the number of elements in the
input sequence.

\subsection{Height}\label{sect:height}
Given a height sequence $\I_\H$, to solve the existence problem, we use the
fact that, from the definition of height of a node, if a particular value $h$ exists in the sequence then so would
at least one instance of each value less than that ($h-1, h-2, \dots,
2, 1$).

\begin{definition}[Strand] The strand of a node is a
  maximal path from the node to a leaf.
   \end{definition}

   We divide the given sequence $\I_\H$ into maximal strands by
   choosing greedily from the given sequence. We ensure that the
   maximal length strands are made before we make any of the smaller
   strands. So, first the root (the largest $h\in\I_\H$) will get a
   strand of its own. Then, the next biggest remaining height value
   will get a strand and so on until no elements remain in the
   sequence. This division is going to be unique because for a node to
   get a height value of $h$, it has to have have at least one child
   whose height is $h-1$, which in turn will need $h-2$ and so
   on. Thus, it gives us a necessary condition for the existence of a
   tree: if while making these strands, we get stuck at some point
   then no tree can be constructed.

   Once these strands are constructed, then we connect them together
   to get a tree, if possible. Any strand with its root's height value
   as $h$ can only be attached as a child of a node whose height value
   is at least $h+1$. Thus if we just check if there are enough places
   where strands can be joined then we could answer whether a tree
   exists. This translates into a simple polynomial time algorithm
   stated in the following pseudo-code:
\begin{verbatim}
places:=0
count[i]:= number of nodes with height i
for i from h_max to 0:
    places := places + K*count[i] - count[i-1]
    if places < 0:
        No tree exists
Tree exists
\end{verbatim}
The \verb|places| variable tells us how many places are left at higher
levels after each strand is added to the tree. If at some point,
\verb|places| becomes less than zero then the last strand we added was
an invalid addition. Hence there would be no tree possible. If it is
non-negative throughout then a tree would be possible.

\subparagraph{Sub-classes}

\begin{itemize}
\item 
Complete $K$-ary: Compare if $\I_\H = \I_\H(T^{c})$, where $T^{c}$ is
the complete tree on $|\I_\H|$ nodes.
\item Full $K$-ary tree: If \verb|places| is zero at the end of the loop,
  then a full $K$-ary tree is possible.
\item Degenerate Tree: Exactly one count of each height value must be
  present.
\end{itemize}

\section{Combining sequences}\label{sect:combined}

In this section we look at the problem of existence given sequences in
combination.

\subsection{Synchronized Height and Depth: $\R(\I_{\H\D})$}
The algorithm for solving the $\R(\I_{\H\D})$ problem combines ideas
from the methods to solve the depth and height problems (Sections
\ref{sect:depth} and \ref{sect:height} respectively) and is easy to
verify.

\begin{enumerate}
\item While finding maximal strands, ensure that depth values are also
  assigned in order.
\item Before computing \verb|places|, put the roots of the strands at
  the correct depth.
\item At every level, check if dedicated possible parent exists for
  each strand rooted at that level. Do this in descending order of the
  root's height values.
\item If one gets stuck at any point in the algorithm, then no tree
  exists, otherwise, a tree exists.
\end{enumerate}

\subparagraph{Asynchronized Height and Depth: $\R(\I_\H,\I_\D)$}
We do not know any solution to the asynchronized version of the height
and depth sequences nor do we have a proof for \npc{}ness.

We have seen that realization given height and depth sequences be
accomplished in polynomial time. We now attempt to solve the existence
problem given subtree sizes by additionally providing depth and/or
height sequences. Note that the existence of complete and degenerate
trees can still be polynomially answered given any combination of
depth and height sequences along with subtree sizes sequence (Section
\ref{sect:sssk-positive}, \ref{sect:depth} and
\ref{sect:height}). 

\begin{corollary}\label{coro:SD}
$\R(\I_\S,\I_\D)$ and $\R(\I_{\S\D})$ for k-ary as well as full k-ary
trees are \snpc{}.
\end{corollary}\begin{proof}
  We show that the $\R(\I_\S)$ instance created during the reduction
  in Section~\ref{sect:prove-sssk} implicitly constructs instances of
  the $\R(\I_\S,\I_\D)$ and $\R(\I_{\S\D})$ problems.

  We note that the depth of the levels in $G$, are
  equal to the subscripts (0 to $d-2$) that are used. For $P$ and $C$
  the depths are $d-1$ and $d$, respectively. The depths for the complete
  trees created in $D$ can be computed along with the subtree
  sizes. Adding this information along with the subtree sizes
  information during reduction allows one to prove the
  $\R(\I_\S,\I_\D)$ and $\R(\I_{\S\D})$ problems to be
  \snpc{}. Similar modifications to corollary~\ref{coro:full-sssk}
  prove the \snpc{}ness for full k-ary trees as well.
\end{proof}

\begin{corollary}\label{coro:HS}
$\R(\I_\S,\I_\H)$ and $\R(\I_{\S\H})$ for k-ary as well as full k-ary
trees are \snpc{}.
\end{corollary}
\begin{proof}
  We show that the $\R(\I_\S)$ instance created during the reduction
  in Section~\ref{sect:prove-sssk} implicitly constructs instances of
  the $\R(\I_\S,\I_\H)$ and $\R(\I_{\S\H})$ problems.

  To know the height of a node, the heights of the children need to be
  known. Due to the way in which the child component is constructed
  (Eq.~\ref{eq:sssk-C}), a child component element $c$ will always
  satisfy the following inequality $K^h \leq c < K^{h+1}$ where
  $h=\lceil\log(M_2)\rceil$. This along with the fact that the
  descendant component attaches a complete
  subtree to $c_i$ ensures that the heights of all $c_i$ will be fixed at
  $h$. The parent and grandparent components' heights are decided
  based on the heights of the $c_i$, hence the heights of all the
  nodes in $S$ can be computed during the transformation. Similarly
  with corollary~\ref{coro:full-sssk} for full k-ary versions.
\end{proof}

\begin{corollary}\label{coro:HDS}
Both the synchronized and asynchronized sequences of heights, depths
and subtree sizes are NP-Complete. (Obvious from Corollaries
\ref{coro:SD} and \ref{coro:HS}.)
\end{corollary}

Since the height and depth sequences are not enough to provide a
solutions, we look at providing the inorder traversal rank (the position of a
node during inorder traversal) synchronized with subtree sizes; we
denote this as $\I_{\S,\ITR}$ and the problem as
$\R(\I_{\S,\ITR})$. While this can be extended to k-ary trees, it is
most easily illustrated in binary trees where the inorder traversal
rank is the rank of the node as it would have been in a binary search
tree. Since we know the root, it allows for trivial partitioning of
the remaining elements into either the left or the right
subtree. After that step, one is left with two smaller problems. This
would allow one to use a divide-and-conquer method to get a polynomial
time solution; giving us the following corollary:

\begin{corollary}\label{coro:itr}
  $\R(\I_{\S,\ITR})$ can be solved in polynomial time.
\end{corollary}
  
The $\R(\I_{\S,\ITR})$ problem throws light into the structural
difficulty in the \sssk{} problem. In the \sssk{} problem, say for
binary trees, deciding the root node's subtree size is obvious (the
largest value in the sequence, say $r$). After that, the second
largest number, say $l$, is necessarily one of the children (let it be
the left child) and $r-l$ will be the remaining child of the root
node. So, the root and its two children can be easily decided (if
they exist) but after that, partitioning the remaining nodes into the
left or right subtree of the root is a difficult task. This is like
the NP-Complete \textsc{partition} problem but is more complex since not only
does the partition have to sum up to a particular value, the partition
must also be realizable as a tree.

Adding the inorder traversal ranks allows definitive partitioning of
the nodes into left/right subtrees. This is then no longer similar to
the \textsc{partition} problem, leading to a polynomial time
solution. This shows that the \textsc{partition} like nature of the
\sssk{} problem is an important part of what makes it \snpc{}. We also
note that whenever we are able to solve a variant of the $\R(\I_\S)$,
we are always creating a unique tree. As soon as there is structural
flexibility in the construction of the tree, the problem becomes
difficult.

\section{Conclusion}\label{sect:conclusion}

In this paper we look at the problem of the existence of rooted k-ary
trees given some combination of sequences of attributes like subtree
sizes, heights and depths. We prove the problem of the existence of
the tree given the subtree sizes sequence to be \snpc{}; problems that
additionally provide height and/or depth sequences, in either
synchronized or asynchronized manner still have the same
complexity. We also prove that in each of these cases, when asked
about the existence of full k-ary trees, the problem remains
\snpc{}. Existence of trees given $\I_\H$, $\I_\D$ and $\I_{\H,\D}$
can be solved in polynomial time, even the problem of existence of
full k-ary trees. For all of these problems existence of complete and
degenerate trees can be solved in polynomial time.

Apart from these, when the inorder traversal rank is given
synchronized with the subtree sizes sequence, the existence of a tree
can be answered in polynomial time. We argued that this is evidence
that the difficulty in the problem is in its partitioning like nature. We
also argued, by comparing the complexity of complete and degenerate
tree variants with the full tree variant, that the uncertainty or
freedom in the structure of a tree plays a role in the intractability
of the problem.

There are many areas for future work. For problems related to subtree
size sequences, exact exponential, approximation algorithms or other
such strategies can be
searched for, which allow solving it more efficiently. One could also look for minimal super-sequences or
maximal sub-sequences which realize a tree. The contrasting nature of
the subtree sizes, height and depth attributes can be studied to gain
a better understanding of the relation between them in terms of
realizability. Along the same lines, we have not been able to find
either an algorithm or a proof of \npc{}ness for the asynchronized
height and depth sequence problem. Studying this problem might throw
insight into how the height and depth of k-ary trees are related. Our
focus throughout the paper has been on rooted k-ary trees but all
these problems can also be asked for general rooted trees.

\subparagraph*{Acknowledgments.}  The author would like to thank
Professor Rahul Muthu and Professor Srikrishnan Divakaran for their
continued guidance and for their helpful suggestions. The author also
thanks Professor Jayanth Varma for his review of the draft.





\bibliography{final-lipics-arxiv}

\end{document}